\documentclass[a4paper]{article}
\usepackage[margin=3cm]{geometry}
\usepackage{amsthm}
\usepackage{enumerate}
\usepackage{mathtools}
\usepackage{latexsym}
\usepackage{amsmath,amssymb,amsfonts}
\usepackage{color}
\usepackage{bbold}
\usepackage{xspace}
\usepackage{lettrine}
\usepackage{fancybox}

\newtheorem{theorem}{Theorem}
\newtheorem{lemma}{Lemma}
\newtheorem{corollary}{Corollary}
\newtheorem{remark}{Remark}

\newcounter{alg}

\newcommand{\eps}{\ensuremath{\varepsilon}}

\usepackage{xspace}
\usepackage{tcolorbox}

\newcommand{\rpay}[2]{\ensuremath{\mathcal{R}(#1,#2)}\xspace}
\newcommand{\cpay}[2]{\ensuremath{\mathcal{C}(#1,#2)}\xspace}
\newcommand{\rreg}{\ensuremath{\text{reg}_r}\xspace}
\newcommand{\creg}{\ensuremath{\text{reg}_c}\xspace}
\newcommand{\rbr}[1]{\ensuremath{\text{B}_r}(#1)\xspace}
\newcommand{\cbr}[1]{\ensuremath{\text{B}_c}(#1)\xspace}

\newcommand{\simplex}{\ensuremath{\Delta^n}\xspace}
\newcommand{\primal}{\text{Primal}\xspace}
\newcommand{\dual}{\text{{Dual}}\xspace}
\newcommand{\xbf}{\ensuremath{\mathbf{x}}\xspace}
\newcommand{\ybf}{\ensuremath{\mathbf{y}}\xspace}
\newcommand{\xbfs}{\ensuremath{\mathbf{x}_s}\xspace}
\newcommand{\ybfs}{\ensuremath{\mathbf{y}_s}\xspace}
\newcommand{\zbf}{\ensuremath{\mathbf{z}}\xspace}
\newcommand{\wbf}{\ensuremath{\mathbf{w}}\xspace}

\newcommand{\tvar}{\ensuremath{t_r}\xspace}
\newcommand{\vvar}{\ensuremath{v_r}\xspace}
\newcommand{\tvarc}{\ensuremath{t_c}\xspace}
\newcommand{\vvarc}{\ensuremath{v_c}\xspace}

\newcommand{\TS}{\ensuremath{\mathtt{TS}}\xspace}
\newcount\Comments  
\definecolor{darkgreen}{rgb}{0,0.6,0}
\newcommand{\kibitz}[2]{\ifnum\Comments=1{\color{#1}{#2}}\fi}

\title{A Polynomial-Time Algorithm for 1/3-Approximate Nash Equilibria in Bimatrix Games}
\author{Argyrios Deligkas$^1$, Michail Fasoulakis$^{2,3}$, Evangelos Markakis$^3$\\
Royal Holloway, University of London$^1$\\
Foundation for Research and Technology-Hellas (FORTH)$^2$\\
Athens University of Economics and Business$^3$}
\date{}
\begin{document}

\maketitle
\begin{abstract}
Since the celebrated PPAD-completeness result for Nash equilibria in bimatrix games, a long line of research has focused on polynomial-time algorithms that compute $\varepsilon$-approximate Nash equilibria. 
Finding the best possible approximation guarantee that we can have in polynomial time has been a fundamental and non-trivial pursuit on settling the complexity of approximate equilibria. Despite a significant amount of effort, the algorithm of Tsaknakis and Spirakis \cite{TS08}, with an approximation guarantee of $(0.3393+\delta)$, remains the state of the art over the last 15 years. 
In this paper, we propose a new refinement of the Tsaknakis-Spirakis algorithm, resulting in a polynomial-time algorithm that computes a $(\frac{1}{3}+\delta)$-Nash equilibrium, for any constant $\delta>0$. The main idea of our approach is to go beyond the use of convex combinations of primal and dual strategies, as defined in the optimization framework of \cite{TS08}, and enrich the pool of strategies from which we build the strategy profiles that we output in certain bottleneck cases of the algorithm. 
\end{abstract}

\newpage
\section{Introduction}
The notion of {\em Nash equilibrium} has been undoubtedly a fundamental solution concept in strategic games, ever since the seminal result of Nash~\cite{Nas51}, on the existence of equilibria for all finite games. Nash's theorem however is only {\em existential}; it only shows that such an equilibrium always exists, but it does not provide an efficient algorithm to find one. In fact, many years after the work of Nash, in a series of breakthrough results, it was proven that computing a Nash equilibrium is PPAD-complete~\cite{DGP06}, even for {\em bimatrix} games~\cite{CDT09}, which provides strong evidence that computing an equilibrium is an intractable problem.

These negative results have naturally led to the study of {\em approximate} Nash equilibria. In an \eps-approximate Nash equilibrium (\eps-NE), no player can increase her payoff more than \eps, by unilaterally changing her strategy. 
In contrast to exact Nash equilibria, the relaxation to $\epsilon$-NE does admit subexponential algorithms. More precisely, the quasi polynomial-time approximation scheme (QPTAS) of~\cite{LMM03} can find an \eps-NE in time $n^{O(\log{n}/\epsilon^2)}$, for a game with $n$ available pure strategies per player. One can then wonder whether the QPTAS could be improved to a PTAS or even a FPTAS.
Unfortunately this does not seem to be the case, as the result of Chen, Deng, and Teng~\cite{CDT09} already ruled out the existence of an FPTAS, unless PPAD=P. Some years later, in another breakthrough result, Rubinstein~\cite{Rub16-PETH} showed that, assuming the exponential-time hypothesis for PPAD, there exists a very small, yet unspecified, constant $\eps^{\star}$ such that finding an \eps-NE requires quasi polynomial time for every constant $\eps < \eps^*$. This would rule out a PTAS too.

Although it seems unlikely to have a polynomial time algorithm for any $\epsilon>0$, it is still important to identify the best constant $\epsilon$ for which we can have an efficient algorithm. In fact, this has been one of the fundamental questions of algorithmic game theory, that is still unresolved. Soon after the initial PPAD-hardness results of~\cite{CDT09,DGP06}, there was a flourish of works along this direction. Kontogiannis, Panagopoulou, and Spirakis~\cite{KPS06} derived a polynomial-time algorithm for $\eps = 3/4$; Daskalakis, Mehta, and Papadimitriou~\cite{DMP07,DMP09} improved it to $\eps=1/2$ and $\eps \approx 0.382$; Bosse, Byrka, and Markakis~\cite{BBM10} achieved $\eps=0.364$; and finally Tsaknakis and Spirakis~\cite{TS08} attained a bound of $\eps = 0.3393+\delta$, for any constant $\delta>0$. 
Ever since this last work however, the progress on this front has stalled, and the result of Tsaknakis and Spirakis (referred to as the \TS algorithm from now on) remains the state of the art over the last 15 years. It is particularly puzzling that so far, it has remained an open problem to even improve the approximation to $1/3+\delta$ (even though it has been conjectured that such an approximation should be feasible). To make things worse, in the very recent work of \cite{ChenDHLL21-Deng-TS-tight}, it was shown that the \TS algorithm and its analysis are tight.

In order to beat the 0.3393-guarantee of the \TS algorithm, it is instructive to understand first its bottleneck cases. At a high level, we can think of the algorithm as consisting of two phases: the \emph{Descent phase} and the \emph{Strategy-construction phase}. In the Descent phase, it performs ``gradient descent'' on the maximum regret among the two players, i.e., the maximum additional gain that a player can have by a unilateral deviation to another strategy. This process terminates at an approximate ``stationary'' point, i.e., a strategy profile such that any local change does not decrease the value of the maximum regret. 
When we reach a $\delta$-stationary point for some small constant $\delta$, the Strategy-construction phase begins. This phase performs a case analysis, based on certain relevant parameters of the game, and tries to decide which strategy profile to output in each of the five cases that arise.

In doing so, the algorithm has at its disposal the $\delta$-stationary profile, along with a ``dual'' strategy profile (produced by solving the dual of the linear program used in the Descent phase). 
A close inspection reveals that one of these two profiles suffices to guarantee a $(\frac{1}{3}+\delta)$-NE in three out of the five cases. In the remaining two cases, the algorithm outputs a convex combination of the stationary and the dual strategies, and this is where the bottleneck occurs, causing the algorithm to output a $(0.3393+\delta)$-NE.

\paragraph*{Our contribution.}
We improve upon the state of the art and provide a polynomial-time algorithm for computing a $(\frac{1}{3}+\delta)$-NE in bimatrix games, for any constant $\delta>0$. More specifically, we modify sufficiently the \TS algorithm by designing an {\em improved} Strategy-construction phase to handle the problematic cases of \TS. Our main insights in doing so are as follows:
\begin{itemize}
    \item Apart from convex combinations between primal (stationary) and dual strategies, we also consider best response strategies to such convex combinations. Hence, we enrich the pool of strategies, out of which we choose the profile to output in each case. As a result, in the cases where the $\delta$-stationary point or the dual profile (or their combinations) do not have the desired guarantee, we have one of the players use a carefully chosen convex combination between our newly defined strategies and her dual strategy.
    \item We produce a more refined case analysis, that is based on  the values of some new auxiliary parameters (e.g., the quantities $\vvar, \tvar$ and $\hat{\mu}$, defined in Section \ref{sec:algo}). These parameters encode payoff differences or regrets of the players for using specific strategies, and they help us in two ways. First, they are used to obtain improved upper bounds on the maximum regret of the $\delta$-stationary profile (Section \ref{sec:regret_bounds}). Secondly, their values greatly help us in decomposing our analysis into convenient subcases in order to establish the approximation guarantee.
\end{itemize}

\paragraph*{Further related work.}
A different notion of approximation of NE is that of \eps-{\em well-supported} NE (\eps-WSNE). In an \eps-WSNE every player is required to place positive probability only to actions that are within $\eps$ of being best responses. Hence, \eps-WSNE are more constrained than \eps-NE, where the players can place a positive probability on any strategy. After a series of papers on the topic~\cite{KS-wsne,FGSS16-approximate-WSNE}, the currently best approximation is for $\eps = 0.6528$ due to~\cite{czumaj2019distributed-wsne}. 

Another line of research has focused on more structured classes of bimatrix games such as: constant-rank games, where the matrix defined by the sum of the two payoff matrices has constant rank~\cite{adsul2011rank1,kannan2010-rank-games,Mehta18-constant-rank};
win-lose games, where the payoff for every pure action is either 0 or 1~\cite{chen2007-win-lose,codenotti2005-win-lose,liu2018-win-lose}; sparse games, where there are only ``a few'' outcomes that yield a non-zero payoff for each player~\cite{chen2006-sparse}, imitation games, where the payoff matrix for one of the players is the identity matrix~\cite{mclennan2010-imitation,mclennan2010simple-imitation,MurhekarMehta20-imitation}; random games, where the payoff entries are drawn from certain distributions~\cite{barany2007nash-random,panagopoulou2014random}; symmetric games, where the payoff matrix of one player is the transpose of the other~\cite{CFJ14,KS11}. 
In most of these classes, it has been possible to obtain improved approximation guarantees and have a better understanding of how to construct approximate equilibria.

Concerning quasi-polynomial algorithms, in addition to the QPTAS of~\cite{LMM03}, three new QPTASs have been obtained, which contain the original result of~\cite{LMM03} as a special case:~\cite{Barman18-caratheodory-qptas} gave a refined, parameterized, approximation scheme;~\cite{BabichenkoBP17} gave a QPTAS that can be applied to multi-player games as well;~\cite{DeligkasFMS22-eps-ETR} gave a more general approach for approximation schemes for the existential theory of the reals.
More recently, more negative results for \eps-NE were derived:~\cite{KothariM18-sum-of-squares} gave an unconditional lower bound, based on the sum of squares hierarchy;~\cite{BoodaghiansBHR20-Smoothed-2Nash} proved PPAD-hardness in the smoothed analysis setting;~\cite{BravermanKW15-QP-LB,DeligkasFS18-constrained-QP-LB,austrin2011inapproximability} gave quasi-polynomial time lower bounds for constrained \eps-NE, under the exponential time hypothesis.

\section{Preliminaries}
\label{sec:prelims}
In what follows, let $[n] := \{1, 2, \ldots, n\}$ and let $\Delta^n$ denote the $(n-1)$-dimensional simplex.  
We focus on $n \times n$ bimatrix games, where $n$ denotes the number of available pure strategies per player. Such games are defined by a pair $(R,C) \in [0,1]^{n \times n}$ of two matrices: $R$ and $C$ are the payoff matrices for the {\em row} player and the {\em column} player respectively. 
We follow the usual assumption in the relevant literature that the matrices are normalized, so that all entries are in $[0, 1]$.
It is also assumed without loss of generality, that both players have the same number of pure strategies, since otherwise one can add dummy strategies to equalize the rows and columns.
The semantics of the payoff matrices are that when the row player picks a row $i \in [n]$ and the column player picks a column $j \in [n]$, then they receive a payoff of $R_{ij}$ and $C_{ij}$ respectively. 

A {\em mixed strategy} is a probability distribution over $[n]$. We use $\xbf \in \simplex$ to denote a mixed strategy for the row player and $x_i$ to denote the probability the player assigns to the pure strategy $i$. For the column player, we use $\ybf \in \simplex$ and $y_i$, respectively. 
If \xbf and \ybf are mixed strategies for the row and the column player respectively, then we call $(\xbf,\ybf)$ a {\em (mixed) strategy profile}. It is often also convenient to represent pure strategies as vectors. Hence, we will use the vector $e_i$, which has 1 at index $i$ and zero elsewhere, to denote the $i$-th pure strategy, in other words the distribution where a player assigns probability one to play the pure strategy $i$.

Given a strategy profile $(\xbf,\ybf)$, the expected payoff of the row player is $\rpay{\xbf}{\ybf}:= \xbf^TR\ybf$, and the expected payoff of the column player is $\cpay{\xbf}{\ybf}:=\xbf^TC\ybf$. Hence, for a pure strategy $e_i$, the term $\rpay{e_i}{\ybf}:=\sum_j  R_{ij}y_j$, denotes the expected payoff of the row player, when she plays the pure strategy $i$ against strategy \ybf of the column player. Similarly, $\cpay{\xbf}{e_j}$ is the expected payoff of the column player when she plays the pure strategy $j$ against \xbf. We say that a pure strategy is a {\em best-response strategy} for a player if it maximizes her expected payoff against a chosen strategy of her opponent. So, under a strategy profile $(\xbf, \ybf)$, the set of pure best responses  for the row player is 
$\rbr{\ybf} := \{i \in [n]: \rpay{e_i}{\ybf} = \max_{i' \in [n]}\rpay{e_{i'}}{\ybf}\}$,
and for the column player, it is 
$\cbr{\xbf} := \{j \in [n]: \cpay{\xbf}{e_j} = \max_{j' \in [n]}\cpay{\xbf}{e_{j'}}\}$.

The {\em regret} of the row player at a profile $(\xbf, \ybf)$, is $\rreg(\xbf,\ybf) = \max_i\rpay{e_i}{\ybf} - \rpay{\xbf}{\ybf}$ and the regret of the column player is $\creg(\xbf,\ybf) = \max_j\cpay{\xbf}{e_j} - \cpay{\xbf}{\ybf}$. The strategy profile $(\xbf,\ybf)$ is an {\em \eps-Nash equilibrium}, or \eps-NE, if the regret of both players is bounded by $\eps \in [0,1]$, formally $\max\{\rreg(\xbf,\ybf), \creg(\xbf,\ybf)\} \leq \eps$. If $\eps = 0$, then the strategy profile $(\xbf,\ybf)$ is an exact Nash equilibrium.

\section{The Tsaknakis-Spirakis algorithm}
\label{sec:TS-algo}
In this section we give a description of the algorithm by \cite{TS08} and we highlight the bottleneck cases, where it fails to provide a $(\frac{1}{3}+\delta)$-approximation. In order to have a self-contained exposition, we also present some of the lemmas that are used in the analysis of \cite{TS08}, which are needed for our work as well.

The core of the algorithm is to consider the function $g(\xbf,\ybf) = \max\{\rreg{(\xbf,\ybf),\creg(\xbf,\ybf)}\}$, i.e., the maximum regret among the two players. 
Clearly, if we arrive at a profile $(\xbf,\ybf)$ such that $g(\xbf,\ybf) \leq \eps$, then $(\xbf,\ybf)$ is an \eps-Nash equilibrium.
At a high level, one can think of \TS as consisting of two phases: the {\emph Descent} phase, and the {\em Strategy-construction} phase.

\medskip
\noindent
{\bf Descent Phase.} During this phase, \TS performs ``gradient descent'' on the function $g(\xbf,\ybf)$, until it reaches a {\it ``stationary'' point}, i.e., a strategy profile such that any local change does not decrease the value of $g$. More concretely, every iteration of the Descent phase performs a series of steps: given the current profile under consideration, it equalizes the regrets of the players, then it solves an appropriate linear program to identify a feasible direction, and finally depending on the solution of the LP, it either updates the strategy profile, or it decides that it has reached an approximate stationary point. 

The first step runs the RegretEqualization procedure described below. This procedure is based on solving a single linear program to equalize the regrets of the two players, and most importantly, it guarantees that the maximum regret does not increase.

\begin{tcolorbox}[title={RegretEqualization$(\xbf_0,\ybf_0)$}]
{\bf Input:} Strategy profile $(\xbf_0,\ybf_0)$.\\
\smallskip
{\bf Output:} A strategy profile $(\xbf, \ybf)$ such that $\rreg(\xbf,\ybf) = \creg(\xbf,\ybf)$.\\
{\bf 1.} If $\rreg(\xbf_0,\ybf_0) \geq \creg(\xbf_0,\ybf_0)$, keep $\ybf_0$ fixed and solve the following linear program\\
\hspace*{10mm} Minimize $\rreg(\xbf,\ybf_0)$\\
\hspace*{10mm} Such that $\rreg(\xbf,\ybf_0) \geq \creg(\xbf,\ybf_0)$ and $\xbf \in \Delta^n$,\\
\smallskip
and return $(\xbf,\ybf_0)$, where $\xbf$ is the solution of the linear program.\\
{\bf 2.} If $\rreg(\xbf_0,\ybf_0) < \creg(\xbf_0,\ybf_0)$, keep $\xbf_0$ fixed and solve the following linear program\\
\hspace*{10mm} Minimize $\creg(\xbf_0,\ybf)$\\
\hspace*{10mm} Such that $\creg(\xbf_0,\ybf) \geq \rreg(\xbf_0,\ybf)$ and $\ybf \in \Delta^n$,\\
and return $(\xbf_0,\ybf)$, where $\ybf$ is the solution of the linear program.\\
\end{tcolorbox}
\noindent
Given the 
output $(\xbf, \ybf)$ of
RegretEqualization, the next step is to either find a feasible direction to follow so as to decrease the maximum regret, or to decide that $(\xbf, \ybf)$ is an approximate stationary point. This is enforced by solving the following linear program.
\begin{tcolorbox}[title={Primal Linear Program: $\primal(\xbf,\ybf)$}]
\begin{tabular}{ r l l }
 minimize & $\gamma$  &  \\ 
 s.t. & $\gamma \geq \rpay{e_i}{\ybf'} - \rpay{\xbf}{\ybf'} - \rpay{\xbf'}{\ybf} +\rpay{\xbf}{\ybf}$, & $\forall i \in \rbr{\ybf}$, \\  
  & $\gamma \geq \cpay{\xbf'}{e_j} - \cpay{\xbf'}{\ybf} - \cpay{\xbf}{\ybf'}+ \cpay{\xbf}{\ybf}$, & $\forall j \in \cbr{\xbf}$,\\
  & $\xbf' \in \simplex, \quad \ybf' \in \simplex$.
\end{tabular}
\end{tcolorbox}
\noindent It is proved in \cite{TS08} that the solution of $\primal(\xbf,\ybf)$ guarantees one of the following:
\begin{enumerate}
    \item \label{point:1} it either identifies a strategy profile $(\xbf',\ybf')$ such that the maximum regret can be strictly decreased by a constant fraction, if we move from $(\xbf, \ybf)$ towards $(\xbf', \ybf')$;
    \item \label{point:2} or it decides that $(\xbf, \ybf)$ is a {\em $\delta$-stationary} point\footnote{This means that the directional derivative of $g(\xbf, \ybf)$ is at least $-\delta$. For the definition of directional derivative, see \cite{TS08}.}, which is the termination criterion of the descent.
\end{enumerate}

\noindent
Putting everything together, the Descent phase of the \TS algorithm is described below, starting from some arbitrary initial strategy profile, and its main properties are captured by the following lemma.
\begin{lemma}[\cite{TS08}]
For any constant $\delta>0$, the Descent phase computes a $\delta$-stationary point, in time polynomial in $1/\delta$ and in the size of the game.
\end{lemma}

\begin{tcolorbox}[title={Descent Phase}]
{\bf Input:} Strategy profile $(\xbf,\ybf)$, a small constant $\delta > 0$ ($\delta<<1/3$).\\
\smallskip
{\bf Output:} A $\delta$-stationary profile $(\xbfs, \ybfs)$ with equal regrets.\\
\smallskip
{\bf 1.} Equalize the regrets of the two players and set $(\xbf,\ybf)\leftarrow$ RegretEqualization$(\xbf,\ybf)$.\\
\smallskip
{\bf 2.} Solve $\primal(\xbf,\ybf)$ and compute $\xbf', \ybf'$, and $\gamma$.\\
\smallskip
{\bf 3.} If $\gamma - g(\xbf,\ybf) \geq -\delta$, set $\xbfs \leftarrow \xbf, \ybfs \leftarrow \ybf$ and stop.\\
\smallskip
{\bf 4.} Else, set 
$\xbf \leftarrow (1-\frac{\delta}{\delta+2})\cdot \xbf + \frac{\delta}{\delta+2}\cdot \xbf'$, $\ybf \leftarrow (1-\frac{\delta}{\delta+2})\cdot \ybf + \frac{\delta}{\delta+2}\cdot \ybf'$ and go to Step 1.
\end{tcolorbox}

\medskip
\noindent
{\bf Strategy-construction Phase.} In this phase, the algorithm utilizes the dual linear program of $\primal(\xbf, \ybf)$, in order to identify some alternative candidate strategies for the players. 
\begin{tcolorbox}[title={Dual Linear Program: $\dual(\xbf,\ybf)$}]
\begin{tabular}{ r l l }
  \text{maximize } & 
    $P \cdot \rpay{\xbf}{\ybf}+ Q \cdot \cpay{\xbf}{\ybf} + a + b$ \\
  $\text{s.t.\ }$ 
  & $p_i \geq 0, \quad i \in \rbr{y}$,\\ 
  & $q_j \geq 0, \quad j \in \cbr{x}$,\\
  & $P = \sum_{i \in \rbr{y}} p_i, \quad 
    Q = \sum_{j \in \cbr{x}} q_j$, \\
  & $P + Q = 1$, \\
  & $a \leq \sum_{i \in \rbr{y}} -\rpay{e_k}{\ybf} \cdot p_i +
    \sum_{j \in \cbr{x}} [-\cpay{e_k}{\ybf}+C_{kj}] \cdot q_j, \quad 1 \leq k \leq n$\\
  & $b \leq \sum_{j \in \cbr{x}} -\cpay{\xbf}{e_l}\cdot q_j +
    \sum_{i \in \rbr{y}} [- \rpay{\xbf}{e_l}+R_{il}] \cdot p_i, \quad 1 \leq l \leq n.$
\end{tabular}
\end{tcolorbox}
%%%%%%%%%%
\noindent
Given the $\delta$-stationary profile $(\xbfs, \ybfs)$ from the Descent phase, the algorithm solves $\dual(\xbfs, \ybfs)$ and computes the following (from the optimal dual variables).
\begin{itemize}
    \item The dual strategy \wbf for the row player, where 
    $w_i = p_i/\sum_{j \in \rbr{\ybfs}}p_j$, for $i\in B_r(\ybfs)$, and $0$ elsewhere; note that by construction, \wbf is a best-response strategy against \ybfs. 
    \item The dual strategy \zbf for the column player, where 
    $z_i = q_i/\sum_{j \in \cbr{\xbfs}}q_j$, for $i\in B_r(\xbfs)$, and $0$ elsewhere; 
    by construction, \zbf is a best-response strategy against \xbfs.
    \item The parameters $P, Q \in [0,1]$, that are useful for the approximation analysis.
\end{itemize}
In addition, we define the following two quantities $\lambda$ and $\mu$, that help in parameterizing the maximum regret bound. These quantities are equal to the payoff difference of a player between the dual and the primal strategies, when the other player uses her dual strategy: 
\begin{align}
\label{eq:lambda-mu}
\lambda = \rpay{\wbf}{\zbf} - \rpay{\xbfs}{\zbf}, \qquad \mu = \cpay{\wbf}{\zbf}-\cpay{\wbf}{\ybfs}.
\end{align}
{\bf Fact.} 
Obviously, $\lambda \leq 1$, and $\mu\leq 1$ and furthermore, $\rpay{\wbf}{\zbf} \geq \lambda$, and $\cpay{\wbf}{\zbf}\geq \mu$. 

\noindent
The algorithm then constructs and outputs a strategy profile as follows.
\begin{tcolorbox}[title={Strategy-Construction Phase}]
{\bf Input:} A $\delta$-stationary strategy profile $(\xbfs,\ybfs)$ from the Descent phase, the dual strategies \wbf, \zbf, and the parameters $\lambda, \mu$.\\
\smallskip
{\bf 1.} If $\min\{\lambda, \mu\} \leq \frac{1}{2}$, then return $(\xbfs,\ybfs)$.\\
\smallskip
{\bf 2.} If $\min\{\lambda, \mu\} \geq \frac{2}{3}$, then return $(\wbf,\zbf)$.\\
\smallskip
{\bf 3.} If $\min\{\lambda, \mu\} > \frac{1}{2}$ and $\max\{\lambda, \mu\} \leq \frac{2}{3}$, then return $(\xbfs,\ybfs)$.\\
%\smallskip
{\bf 4.}~Else if $\lambda \geq \mu$, then return the strategy profile with the minimum regret between  $\left(\frac{1}{1+\lambda-\mu}\cdot \wbf +  \frac{\lambda-\mu}{1+\lambda-\mu}\cdot \xbfs, \zbf \right)$ and $(\xbfs,\ybfs)$.\\
{\bf 5.}~Else if $\lambda < \mu$, then return the strategy profile with the minimum regret between $\left(\wbf, \frac{1}{1+\mu -\lambda}\cdot \zbf +  \frac{\mu-\lambda}{1+\mu-\lambda}\cdot \ybfs \right)$ and $(\xbfs,\ybfs)$.
\end{tcolorbox}

\begin{theorem}[\cite{TS08}]
For any constant $\delta>0$, the \TS algorithm computes in polynomial time a $(0.3393+\delta)$-NE.
\end{theorem}

\begin{remark}
One could also check all the proposed profiles of this phase at every iteration of the Descent phase, as presented in \cite{TS08}, and stop if we have reached already the desired approximation. But this does not affect the worst-case running time, which occurs when the Descent phase terminates at a $\delta$-stationary point. 
\end{remark}
\noindent
We present below some important lemmas from \cite{TS08} that are needed in our analysis too. For the sake of completeness, we provide their proofs here. 

The first, and most important, lemma below shows how $\primal(\xbfs, \ybfs)$ and $\dual(\xbfs, \ybfs)$ can be used to bound the value of the maximum regret, $g(\xbfs, \ybfs)$.

\begin{lemma}[implied by \cite{TS08}]
\label{lem:stationary-bound}
Let $(\xbfs,\ybfs)$ be a $\delta$-stationary point produced by the Descent Phase, for a constant $\delta>0$. Let also $\wbf, \zbf$ and $P$, be derived by an optimal solution to $\dual(\xbfs, \ybfs)$, as seen before. Then, for any strategy profile $(\xbf',\ybf')$, it holds that 
 \begin{align*}
 g(\xbfs,\ybfs) & \leq 
 P\cdot(\rpay{\wbf}{\ybf'} - \rpay{\xbf'}{\ybfs} -\rpay{\xbfs}{\ybf'} +\rpay{\xbfs}{\ybfs}) + \\
  & (1-P)\cdot (\cpay{\xbf'}{\zbf} - \cpay{\xbf'}{\ybfs}-\cpay{\xbfs}{\ybf'} +\cpay{\xbfs}{\ybfs})  +\delta.
\end{align*}
\end{lemma}

\begin{proof}
In~\cite{TS08} and in \cite{DFSS17} it was proven that for any $\delta$-stationary strategy profile, it holds that $g(\xbfs,\ybfs) \leq \gamma + \delta$, where $\gamma$ is the optimal solution of $\primal(\xbfs,\ybfs)$. Hence, in order to prove the lemma it suffices to bound the value of $\gamma$. We will do this by using the dual linear program.
It is easy to see first of all that the primal program is feasible and bounded, since the strategies belong to the simplex and also $\gamma$ is bounded by below when the input of the primal is the $\delta$-stationary point. This means that it has an optimal solution and the same holds for the dual program as well. Therefore, we can apply the LP Duality theorem, and have that for any pair of primal and dual optimal solutions for $\primal(\xbfs,\ybfs)$ and $\dual(\xbfs,\ybfs)$ respectively, the objective functions are equal. This yields:
\begin{align}
\label{eq:lpduality}
    \gamma = a + b + P \cdot \rpay{\xbfs}{\ybfs} + (1-P)\cdot \cpay{\xbfs}{\ybfs}.
\end{align}
In addition, from the constraints of $\dual(\xbfs,\ybfs)$, we have the following two inequalities
\begin{align*}
a & \leq -P \cdot \rpay{e_k}{\ybfs} - (1-P)\cdot\cpay{e_k}{\ybfs} + (1-P)\cdot \cpay{e_k}{\zbf}, \quad \forall k \in [n],\\
b & \leq -(1-P)\cdot \cpay{\xbfs}{e_l} -P\cdot \rpay{\xbfs}{e_l} + P\cdot \rpay{\wbf}{e_l}, \quad \forall l \in [n].
\end{align*}
Hence, since the inequalities above hold for every $k \in [n]$ and every $l \in [n]$, it must be true that for any profile $(\xbf',\ybf')$ it holds that (by multiplying each inequality involving $a$ with $x_k'$ and adding them all up for all $k\in [n]$, and similarly for $b$)  
\begin{align*}
a & \leq -P \cdot \rpay{\xbf'}{\ybfs} - (1-P)\cdot\cpay{\xbf'}{\ybfs} + (1-P)\cdot \cpay{\xbf'}{\zbf},\\
b & \leq -(1-P)\cdot \cpay{\xbfs}{\ybf'} -P\cdot \rpay{\xbfs}{\ybf'} + P\cdot \rpay{\wbf}{\ybf'}.
\end{align*}
The lemma follows by replacing these bounds for $a$ and $b$ in \eqref{eq:lpduality}.
\end{proof}

Lemma~\ref{lem:stationary-bound} plays a crucial role as it allows us to bound $g(\xbfs, \ybfs)$ in terms of $\lambda, \mu$ and $P$, by making appropriate choices for $\xbf'$ and $\ybf'$.
This is used both in the following lemma and in Lemma \ref{lem:stationary_bounds} of Section \ref{sec:algo}.

\begin{lemma}[\cite{TS08}]
\label{lem:TS-lambda-mu-bound}
Let $(\xbfs,\ybfs)$ be a $\delta$-stationary point produced by the Descent phase, for a constant $\delta>0$, and let $P$ be obtained by an optimal solution of $\dual(\xbfs, \ybfs)$. It holds that $g(\xbfs, \ybfs) \leq \min \{P\cdot\lambda, (1-P)\cdot \mu\} + \delta \leq \frac{\lambda\cdot\mu}{\lambda + \mu} + \delta\leq \frac{\lambda + \mu}{4} +\delta$.
\end{lemma}

\begin{proof}
The first inequality follows from Lemma~\ref{lem:stationary-bound}, since:
\begin{itemize}
    \item if we replace $(\xbf',\ybf')$ with $(\xbfs,\zbf)$ in the upper bound of Lemma \ref{lem:stationary-bound}, we get that $g(\xbfs, \ybfs) \leq P \cdot (\rpay{\wbf}{\zbf} - \rpay{\xbfs}{\zbf}) + \delta = P \cdot \lambda + \delta$;
    \item if we replace $(\xbf',\ybf')$ with $(\wbf,\ybfs)$, we get that $g(\xbfs, \ybfs) \leq  (1-P)\cdot (\cpay{\wbf}{\zbf} - \cpay{\wbf}{\ybfs}) + \delta = (1-P)\cdot \mu + \delta$.
\end{itemize}
Notice now that $P \cdot \lambda$ is increasing with $P$, and $(1-P)\cdot \mu$ is decreasing with $P$. Hence, the maximum of the minimum of these two linear functions is attained at the point where they are equal, i.e., for $P' = \frac{\mu}{\lambda+\mu}$ (given also that $P'\in [0, 1]$, which is obviously true). Hence the maximum regret is at most $\frac{\lambda \cdot \mu}{\lambda+\mu} + \delta$.
Finally, it is also easy to see that $\frac{\lambda \cdot \mu}{\lambda + \mu} \leq \frac{\lambda +\mu}{4}$ since $(\lambda-\mu)^2\geq 0 \Rightarrow \lambda^2+\mu^2 \geq 2\lambda\cdot \mu  \Rightarrow \lambda^2+\mu^2 +2\lambda \cdot \mu \geq 4\lambda \cdot \mu \Rightarrow (\lambda+\mu)^2\geq 4 \lambda\cdot \mu \Rightarrow \frac{\lambda \cdot \mu}{\lambda + \mu} \leq \frac{\lambda + \mu}{4}.$ 
\end{proof}

One may worry that the bound $\frac{\lambda \cdot \mu}{\lambda + \mu}$ is not well-defined when $\lambda + \mu =0$. However, as we explain below, this is not a concern.

\begin{corollary}
\label{cor:lambda-mu}
We can assume that both $\lambda> 0$ and $\mu > 0$, otherwise $(\xbfs, \ybfs)$ is a $\delta$-Nash equilibrium. 
\end{corollary}

\begin{proof}
Consider the first bound that was established in Lemma \ref{lem:TS-lambda-mu-bound}, that $g(\xbfs, \ybfs) \leq \min \{P\cdot\lambda, (1-P)\cdot \mu\} + \delta$.
As $P\geq 0$, then if $\lambda \leq 0$,  we would have $g(\xbfs, \ybfs) \leq \delta$. Thus, $(\xbfs, \ybfs)$ would be a $\delta$-Nash equilibrium, and since $\delta$ is a small constant, we would have a $\frac{1}{3}$-Nash equilibrium. In the same manner, we can argue that $\mu > 0$.
\end{proof}

The definitions of $\lambda$ and $\mu$, along with Lemma~\ref{lem:TS-lambda-mu-bound} can immediately be used to prove that Cases 1-3 from the Strategy-construction Phase return a $(\frac{1}{3}+\delta)$-Nash equilibrium. Hence, the bottleneck of the \TS algorithm comes from Cases 4 and 5.
In fact, it was also recently shown in \cite{ChenDHLL21-Deng-TS-tight} that the analysis of these cases in \cite{TS08} is tight, and therefore one needs to come up with a different construction in order to obtain an improvement.

\begin{lemma}[\cite{TS08}]
\label{lem:TS-good-cases}
Cases 1-3 from the Strategy-construction Phase return a $(\frac{1}{3}+\delta)$-Nash equilibrium.
\end{lemma}
\begin{proof}
We will consider every case independently.
\begin{itemize}
    \item If $\min\{\lambda,\mu\}\leq \frac{1}{2}$, by Lemma~\ref{lem:TS-lambda-mu-bound} we have that $g(\xbfs,\ybfs)\leq \frac{\lambda \cdot \mu}{\lambda+\mu}+\delta\leq \frac{\min\{\lambda,\mu\}}{\min\{\lambda,\mu\} +1}+\delta \leq \frac{1/2}{1/2+1}+\delta \leq \frac{1}{3}+\delta$. Here, the second inequality comes from the fact that $\frac{\lambda \cdot \mu}{\lambda + \mu}$ is an increasing function of $\max\{\lambda,\mu\}$, and also $\max\{\lambda,\mu\}\leq 1$. 
    \item If $\min\{\lambda,\mu\} \geq \frac{2}{3}$, then $g(\wbf,\zbf) \leq \max\{1-\rpay{\wbf}{\zbf}, 1-\cpay{\wbf}{\zbf}\}$. But 
since $\rpay{\wbf}{\zbf}\geq \lambda$ and $\cpay{\wbf}{\zbf}\geq \mu$, we have that the regret is at most $1- \min\{\lambda,\mu\}\leq \frac{1}{3}$.
    \item If $\min\{\lambda,\mu\}> \frac{1}{2}$ and $\max\{\lambda,\mu\}\leq \frac{2}{3}$, by Lemma~\ref{lem:TS-lambda-mu-bound} we have $g(\xbfs,\ybfs)\leq \frac{\lambda \cdot \mu}{\lambda + \mu}+\delta \leq \frac{\frac{2}{3}\cdot \frac{2}{3}}{\frac{2}{3} +\frac{2}{3}}+\delta \leq \frac{1}{3}+\delta$, since $\frac{\lambda\cdot \mu}{\lambda +\mu}$ is an increasing function of $\lambda$ and $\mu$.
\end{itemize}
\end{proof}

Thus, in the next section, we will focus on the remaining cases, when $\min\{\lambda,\mu\}\in(\frac{1}{2},\frac{2}{3}]$ and $\max\{\lambda,\mu\}\in(\frac{2}{3},1]$.

\section{Improved Strategy-construction Phase}
\label{sec:algo}

In this section we replace Cases 4 and 5 from the original \TS algorithm in order to bypass the bottleneck in the approximation. To do so, we utilize the $\delta$-stationary point $(\xbfs, \ybfs)$, the dual strategies $\wbf, \zbf$, their convex combinations and best-response strategies to such combinations. 
We then perform a more refined analysis and prove that in every case we can efficiently construct a tailored strategy profile that is a $(\frac{1}{3}+\delta)$-Nash equilibrium.

\medskip
\noindent
Our new Strategy-construction phase works as follows.

\begin{tcolorbox}[title={Improved Strategy-construction Phase}]
{\bf Input:} A $\delta$-stationary point $(\xbfs,\ybfs)$ from the Descent phase, the dual strategies \wbf, \zbf, and the parameters $\lambda, \mu$.\\
\smallskip
{\bf 1.} If $\min\{\lambda, \mu\} \leq \frac{1}{2}$, then return $(\xbfs,\ybfs)$.\\
\smallskip
{\bf 2.} If $\min\{\lambda, \mu\} \geq \frac{2}{3}$, then return $(\wbf,\zbf)$.\\
\smallskip
{\bf 3.} If $\min\{\lambda, \mu\} > \frac{1}{2}$ and $\max\{\lambda, \mu\} \leq \frac{2}{3}$, then return $(\xbfs,\ybfs)$.\\
{\bf 4.} If $\frac{1}{2}< \lambda \leq \frac{2}{3} < \mu$:
\begin{itemize}
\item Set $\hat{\ybf} = \frac{1}{2}\cdot \ybfs +\frac{1}{2}\cdot \zbf$.
\item Find a best response $\hat{\wbf}$ against $\hat{\ybf}$.
\item Set $\tvar = \rpay{\hat{\wbf}}{\hat{\ybf}} - \rpay{\wbf}{\hat{\ybf}}; \quad \vvar = \rpay{\wbf}{\ybfs} - \rpay{\hat{\wbf}}{\ybfs}; \quad \hat{\mu} = \cpay{\hat{\wbf}}{\zbf} - \cpay{\hat{\wbf}}{\ybfs}$.

\begin{itemize}
    \item[{\bf 4.1}] If $\vvar + \tvar \geq \frac{\mu-\lambda}{2}$ and $\hat{\mu}\geq \mu-\vvar-\tvar$, then set $p = \frac{2\cdot(\vvar+\tvar)-(\mu-\lambda)}{2\cdot(\vvar+\tvar)}$ and return the strategy profile with the minimum regret among $(p\cdot\wbf+(1-p)\cdot\hat{\wbf},\zbf)$ and $(\xbfs,\ybfs)$.
    \item[{\bf 4.2}] Else, set $q = \frac{1-\mu/2-\tvar}{1+\mu/2-\lambda - \tvar}$ and return the strategy profile with the minimum regret among $(\wbf,(1-q)\cdot \hat{\ybf} + q \cdot \zbf)$ and $(\xbfs,\ybfs)$.
\end{itemize}
\end{itemize}
{\bf 5.} If $\frac{1}{2}< \mu\leq \frac{2}{3} < \lambda$ (symmetric to Case 4):
\begin{itemize}
    \item Set $\hat{\xbf} = \frac{1}{2}\cdot \xbfs +\frac{1}{2}\cdot \wbf$.
    \item Find a best response $\hat{\zbf}$ against $\hat{\xbf}$.
    \item Set 
    $\tvarc = \cpay{\hat{\xbf}}{\hat{\zbf}} - \cpay{\hat{\xbf}}{\zbf}; \quad 
    \vvarc = \cpay{\xbfs}{\zbf} - \cpay{\xbfs}{\hat{\zbf}}; \quad 
    \hat{\lambda} = \rpay{\wbf}{\hat{\zbf}} - \rpay{\xbfs}{\hat{\zbf}}$.
    
\begin{itemize}
\item[{\bf 5.1}] If $\vvarc + \tvarc \geq \frac{\lambda-\mu}{2}$ and $\hat{\lambda} \geq \lambda - \vvarc - \tvarc$, then set $p = \frac{2\cdot(\vvarc+\tvarc)-(\lambda-\mu)}{2\cdot(\vvarc+\tvarc)}$ and
return the strategy profile with the minimum regret among $(\wbf,p\cdot\zbf+(1-p)\cdot\hat{\zbf})$ and $(\xbfs,\ybfs)$.
\item[{\bf 5.2}] Else, set $q = \frac{1-\lambda/2-\tvarc}{1+\lambda/2-\mu -\tvarc}$ and return the strategy profile with the minimum regret among $((1-q)\cdot \hat{\xbf} + q \cdot \wbf,\zbf)$ and $(\xbfs,\ybfs)$.
\end{itemize}
\end{itemize}
\end{tcolorbox}

\noindent
Note that Cases 1-3 are identical to the Strategy-construction phase of the \TS algorithm. Thus, by Lemma~\ref{lem:TS-good-cases} they return a $(\frac{1}{3}+\delta)$-Nash equilibrium.
The new part concerns Cases 4 and 5.
The analysis in both cases is based on certain auxiliary parameters 
($\vvar, \tvar$ and $\hat{\mu}$ for Case 4 and analogously for Case 5), that we define in the statement of the algorithm. These parameters encode payoff differences or regrets of the players for using specific strategies, and they  help us decompose the problem into convenient subcases, so as to obtain better upper bounds on the maximum regret.

\medskip
\noindent Our main result is as follows:

\begin{theorem}
For any constant $\delta>0$, we can compute in polynomial-time a $(\frac{1}{3}+\delta)$-Nash equilibrium.
\end{theorem}

To prove the theorem, it suffices to analyze Case 4, where $\frac{1}{2}< \lambda\leq \frac{2}{3}<\mu$, since Case 5 is symmetric to Case 4 and is analyzed in exactly the same way.

\medskip
\noindent {\bf Intuition and Roadmap.}
The overall analysis in the sequel looks rather technical, therefore, we will first provide some elaboration on the choices that the algorithm makes in Case 4. The first crucial component in the design of the new algorithm is that the upper bounds on the regret of the $\delta$-stationary point $(\xbfs, \ybfs)$, obtained in Lemma \ref{lem:TS-lambda-mu-bound}, can be further refined based on the values of the parameters $\lambda, \mu, \hat{\mu}, \vvar$. This is precisely implemented in Section \ref{sec:regret_bounds} with Lemmas \ref{lem:stationary_bounds}, \ref{lem:stationary_hatm_1}, and \ref{lem:stationary_hatm2}. Once this is done, we then try to answer the following question: Whenever $(\xbfs, \ybfs)$ does not provide a $(\frac{1}{3}+\delta)$-approximation, which profiles can form alternative candidates for a better performance? One idea is to exploit the dual strategies $\wbf$, and $\zbf$, as was also done in \cite{TS08}. However, the profile $(\wbf, \zbf)$ may not be a $(\frac{1}{3}+\delta)$-equilibrium either (in most cases). A next attempt then is to consider appropriate convex combinations of the primal and the dual strategy for each player, i.e., a combination of $\xbfs$ and $\wbf$ for the row player and $\ybfs$ and $\zbf$ for the column player. Unfortunately, this again does not work in all cases. But one next step is to also take into consideration best-response strategies against such convex combinations. E.g., the strategy $\hat{\wbf}$ defined in Case 4 is a best response to the equiprobable combination of $\ybfs$ and $\zbf$. This completes our weaponry, and at the end, in all subcases of Case 4, we consider profiles where the row player uses a convex combination of $\wbf$ and $\hat{\wbf}$, and the column player selects a combination between her primal and dual strategies, $\ybfs$ and $\zbf$. Analogous profiles with the roles of the players reversed are constructed for Case 5 too. Finally, we also know that whenever $(\xbfs, \ybfs)$ does not attain a $(\frac{1}{3}+\delta)$-approximation, this restricts the relation between the parameters $\lambda$, $\mu$, $\hat{\mu}$ and $\vvar$ due to the lemmas of Section \ref{sec:regret_bounds}. This is exploitable for us in the sense that it allows us to construct the exact coefficients for the convex combinations that we use so as to have the desired approximation.  

To proceed, we start with two helpful observations, which are used repeatedly for the analysis of Cases 4.1 and 4.2. 

\begin{lemma}
\label{lem:rwz-bound}
It holds that $\rpay{\hat{\wbf}}{\zbf} \geq \lambda +\vvar+2\tvar$.
\end{lemma}
\begin{proof}
By the definition of $\tvar$, inside Case 4, we have that it holds that $\rpay{\hat{\wbf}}{\hat{\ybf}} = \rpay{\wbf}{\hat{\ybf}} + \tvar$. Hence,
\[
\frac{\rpay{\hat{\wbf}}{\ybfs}}{2} + \frac{\rpay{\hat{\wbf}}{\zbf}}{2} = \rpay{\hat{\wbf}}{\hat{\ybf}} = \frac{\rpay{\wbf}{\ybfs}}{2} + \frac{\rpay{\wbf}{\zbf}}{2} +\tvar \Rightarrow \]
\[\rpay{\hat{\wbf}}{\zbf} = (\rpay{\wbf}{\ybfs}- \rpay{\hat{\wbf}}{\ybfs}) + \rpay{\wbf}{\zbf} + 2\tvar \geq \lambda +\vvar+2\tvar, \]
since $\rpay{\wbf}{\ybfs}-\rpay{\hat{\wbf}}{\ybfs} = \vvar$, and $\rpay{\wbf}{\zbf} \geq \lambda$ (by the fact after Equation \eqref{eq:lambda-mu}).
\end{proof}

\begin{corollary}
\label{cor:vvar-rvar-bounds}
It holds that $\vvar \leq 1-\lambda -2\tvar$, or equivalently $\tvar \leq \frac{1-\lambda -\vvar}{2}$.
\end{corollary}
\begin{proof}
By the previous lemma we have $\rpay{\hat{\wbf}}{\zbf} \geq \lambda +\vvar+2\tvar  \Rightarrow \vvar \leq 1-\lambda -2\tvar$, since $\rpay{\hat{\wbf}}{\zbf} \leq 1$.
\end{proof}

\subsection{Bounding the regret of $\delta$-stationary points}
\label{sec:regret_bounds}

In this subsection, we provide three crucial lemmas that provide different ways of bounding the maximum regret of any $\delta$-stationary point. The first of these lemmas is an improvement over \cite{TS08}, where we add a third upper bound for the $\delta$-stationary point, in addition to the bounds stated in Lemma \ref{lem:TS-lambda-mu-bound} from Section \ref{sec:TS-algo}. 

\begin{lemma}
\label{lem:stationary_bounds}
Let $(\xbfs,\ybfs)$ be a $\delta$-stationary point with $\delta\geq 0$, and let $P$ be obtained by an optimal solution of $\dual(\xbfs, \ybfs)$, as the sum of the dual variables: $P = \sum_{i\in B_r(\ybfs)} p_i$. It holds that $g(\xbfs,\ybfs) \leq \min\{P\cdot \lambda, (1-P)\cdot \mu, P\cdot \vvar + (1-P)\cdot \hat{\mu}\} + \delta$. 
\end{lemma}
\begin{proof}
By Lemma \ref{lem:TS-lambda-mu-bound} it holds that $g(\xbfs,\ybfs)\leq \min\{P\cdot\lambda,(1-P)\cdot\mu\} + \delta$. So, it suffices to prove that $g(\xbfs,\ybfs)\leq P\cdot\vvar+(1-P)\cdot\hat{\mu} + \delta$. This follows from Lemma~\ref{lem:stationary-bound} when we set $(\xbf', \ybf') = (\hat{\wbf},\ybfs)$. Indeed, in this case we have $g(\xbfs,\ybfs)\leq P \cdot (\rpay{\wbf}{\ybfs} - \rpay{\hat{\wbf}}{\ybfs} - \rpay{\xbfs}{\ybfs} +\rpay{\xbfs}{\ybfs}) + (1-P)\cdot (\cpay{\hat{\wbf}}{\zbf} -\cpay{\hat{\wbf}}{\ybfs} -\cpay{\xbfs}{\ybfs}+\cpay{\xbfs}{\ybfs}) = P \cdot \vvar + (1-P)\cdot \hat{\mu} +\delta$, by the definitions of $\vvar$ and $\hat{\mu}$.
\end{proof}

The remaining two lemmas help in attaining a more fine-grained analysis on upper bounding the regret of the players, under the restrictions on the values of $\lambda$ and $\mu$ in Case 4.
\begin{lemma}
\label{lem:stationary_hatm_1}
Let $(\xbfs,\ybfs)$ be a $\delta$-stationary point with $\delta\geq 0$, and let $\hat{\mu}\geq \vvar$, and $\lambda > \frac{1}{2}$. Then, it holds that $g(\xbfs,\ybfs) \leq \frac{\hat{\mu}\cdot \lambda }{\lambda + \hat{\mu} - \vvar} +\delta$.
\end{lemma}
\begin{proof}
By Lemma \ref{lem:stationary_bounds} we have $g(\xbfs,\ybfs) \leq \min\{P \cdot \lambda, P\cdot \vvar + (1-P)\cdot \hat{\mu}\}+\delta$. 
Note that $P\cdot\lambda$ is an increasing linear function of $P$ and $P\cdot \vvar + (1-P)\cdot \hat{\mu}$ is a decreasing linear function of $P$, because $\hat{\mu}\geq \vvar$. Therefore, the maximum of the minimum of these two functions is achieved at the point where they are equal, which is for $P' = \frac{\hat{\mu}}{\lambda + \hat{\mu}-\vvar}$, as long as $P'\in [0, 1]$ (recall that $P$ is constrained to belong to this interval). To check that $P'$ is a valid point, observe first that since $\lambda > \frac{1}{2}$ and $\hat{\mu}\geq \vvar$, the denominator of $P'$ is positive. Also, again using that $\lambda > \frac{1}{2}$, Corollary \ref{cor:vvar-rvar-bounds} implies that $\vvar\leq 1-\lambda\leq \frac{1}{2}$, hence $\lambda > \vvar$, which means that $P'\in[0, 1]$. 
Thus, $g(\xbfs,\ybfs)\leq \lambda \cdot P' + \delta = \frac{\hat{\mu}\cdot\lambda}{\lambda + \hat{\mu}-\vvar}+\delta$.
\end{proof}

\begin{lemma}
\label{lem:stationary_hatm2}
Let $(\xbfs,\ybfs)$ be a $\delta$-stationary point with $\delta\geq 0$, and let $\hat{\mu}< \vvar$, $\lambda > \frac{1}{2}$, and $\mu>\frac{2}{3}$. Then, it holds that $g(\xbfs,\ybfs) \leq \frac{\vvar\cdot \mu }{\mu - \hat{\mu} + \vvar} +\delta$.
\end{lemma}
\begin{proof}
By Lemma \ref{lem:stationary_bounds}, we have $g(\xbfs,\ybfs) \leq \min\{(1-P) \cdot \mu, P\cdot \vvar + (1-P)\cdot \hat{\mu}\}+\delta$. 
In analogy to Lemma \ref{lem:stationary_hatm_1}, we have one linear increasing function of $P$ and one linear decreasing function. Hence, the maximum of the minimum of these functions is attained at the point where they are equal, which is for $P' = \frac{\mu-\hat{\mu}}{\mu - \hat{\mu}+\vvar}$, as long as $P'\in [0, 1]$.
By the assumptions on $\lambda$, $\mu$, and by Corollary \ref{cor:vvar-rvar-bounds}, since $\vvar>\hat{\mu}$, we have $\mu > \frac{2}{3}> \vvar >\hat{\mu}$. Hence $P'\in [0,1]$. The final bound we obtain is $g(\xbfs,\ybfs)\leq (1-P')\cdot \mu  + \delta = \frac{\vvar \cdot \mu}{\mu - \hat{\mu}+\vvar}+\delta$.
\end{proof}

\subsection{Case 4.1 of the Improved Strategy-construction Phase}

We now analyze the approximation we obtain, when we fall into Case 4.1 of the algorithm. We establish that either the $\delta$-stationary point has the desired approximation or otherwise, this is achieved by having the row player use an appropriate convex combination of $\wbf$ and $\hat{\wbf}$ and the column player play the dual strategy $\zbf$.

\begin{lemma}
If $\vvar+\tvar \geq \frac{\mu-\lambda}{2}$, and $\hat{\mu} \geq \mu -\vvar-\tvar$, then for the strategy profile $(p\cdot\wbf+(1-p)\cdot\hat{\wbf},\zbf)$, with $p = \frac{2\cdot(\vvar+\tvar)-(\mu-\lambda)}{2\cdot(\vvar+\tvar)}$, the payoff of both the row and the column player is at least $\frac{\lambda+\mu}{2}$.  
\end{lemma}
\begin{proof}
Note first that under the assumptions of the lemma, and since $\mu> \lambda$, the parameter $p$ is a valid probability. 
For the row player, we have that her payoff is 
\begin{align*}
\rpay{p\cdot\wbf+(1-p)\cdot\hat{\wbf}}{\zbf} & = p\cdot \rpay{\wbf}{\zbf} + (1-p)\cdot\rpay{\hat{\wbf}}{\zbf}\\
 & \geq p\cdot\lambda +(1-p)\cdot\lambda +(1-p)\cdot(\vvar+\tvar) \qquad (\text{from Lemma~\ref{lem:rwz-bound}})\\
 & = \lambda +(1-p)\cdot(\vvar+\tvar) \\
 & = \lambda +\frac{(\mu-\lambda)}{2} \qquad \left(\text{since $1-p = \frac{\mu-\lambda}{2\cdot(\vvar+\tvar)}$}\right)\\
 & = \frac{\lambda + \mu}{2}.
\end{align*}

For the column player we have that her payoff is
\begin{align*}
\cpay{p\cdot\wbf+(1-p)\cdot\hat{\wbf}}{\zbf} & = p\cdot\cpay{\wbf}{\zbf}+(1-p)\cdot\cpay{\hat{\wbf}}{\zbf}\\
& \geq p\cdot \mu +(1-p)\cdot \hat{\mu} \quad (\text{Since $\hat{\mu} = \cpay{\hat{\wbf}}{\zbf} - \cpay{\hat{\wbf}}{\ybfs} \leq \cpay{\hat{\wbf}}{\zbf}$})\\
& \geq p\cdot\mu +(1-p)\cdot\mu -(1-p)\cdot(\vvar+\tvar) \quad \text{(Since $\hat{\mu}\geq \mu -\vvar -\tvar$)}\\
& = \mu - \frac{(\mu-\lambda)}{2} \quad \left(\text{since $1-p = \frac{\mu-\lambda}{2\cdot(\vvar+\tvar)}$}\right)\\
& = \frac{\mu+\lambda}{2}.
\end{align*}
\end{proof}

\begin{lemma}
Let $p\in [0, 1]$, be such that $\rpay{p \cdot \wbf +(1-p) \cdot \hat{\wbf}}{\zbf}\geq \frac{\lambda + \mu}{2}$, and  $\cpay{p \cdot \wbf +(1-p) \cdot \hat{\wbf}}{\zbf} \geq \frac{\lambda + \mu}{2}$. Then, either $(\xbfs,\ybfs)$ is a $(\frac{1}{3}+\delta)$-Nash equilibrium, or $(p \cdot \wbf +(1-p) \cdot \hat{\wbf},\zbf)$ is a $\frac{1}{3}$-Nash equilibrium.
\end{lemma}
\begin{proof}
The regret of either player at the strategy profile $(p \cdot \wbf +(1-p) \cdot \hat{\wbf},\zbf)$ is at most $1-\frac{\lambda+\mu}{2}$, since the payoff of any player is no less than $\frac{\lambda+\mu}{2}$ and the best-response payoff is at most 1. On the other hand, by Lemma \ref{lem:TS-lambda-mu-bound} the regret of each player at the $\delta$-stationary point $(\xbfs,\ybfs)$ is at most $\frac{\lambda+\mu}{4}+\delta$. Thus, if $\lambda + \mu \leq \frac{4}{3}$, then $g{(\xbfs,\ybfs)}\leq \frac{1}{3}+\delta$. Otherwise, the maximum regret at the profile $(p \cdot \wbf +(1-p) \cdot \hat{\wbf},\zbf)$ is at most $1-\frac{\lambda+\mu}{2} \leq \frac{1}{3}$.
\end{proof} 

\subsection{Case 4.2 of the Strategy-construction phase}
In this case it holds that either $\vvar+\tvar<\frac{\mu-\lambda}{2}$ or $\hat{\mu}<\mu-\vvar-\tvar$. It turns out that this is a technically more intriguing case, and the reason is that the parameters are less constrained, compared to Case 4.1. As a result, we need to consider different subcases in order to have tighter upper bounds.
We recall that the algorithm in this case outputs either $(\xbfs, \ybfs)$, or a profile where the row player selects her dual strategy $\wbf$, which is a best response against $\ybfs$, and the column player plays a convex combination between $\hat{\ybf}$ and $\zbf$, which by the definition of $\hat{\ybf}$, is a convex combination of her primal strategy $\ybfs$ and her dual strategy $\zbf$. 

\begin{lemma}
\label{regrets of strategy profile}
The regret of the row player at $(\wbf,\hat{\ybf})$ is $\tvar$ and the regret of the column player is at most $1-\frac{\mu}{2}$.
\end{lemma}
\begin{proof}
By definition, the regret of the row player is $\tvar$, since $\hat{\wbf}$ is a best-response strategy against $\hat{\ybf}$.
On the other hand, recall by the definition of $\mu$, that $\cpay{\wbf}{\zbf}\geq \mu$. 
So, we have that $\cpay{\wbf}{\hat{\ybf}} = \frac{\cpay{\wbf}{\ybfs}}{2} + \frac{\cpay{\wbf}{\zbf}}{2} \geq \frac{\cpay{\wbf}{\zbf}}{2} \geq \frac{\mu}{2}$. Thus, since the maximum payoff is less than or equal to 1, we have that the regret of the column player is at most $1-\frac{\mu}{2}$.
\end{proof}

We now quantify the regret of the players at the profile $(\wbf,(1-q)\cdot \hat{\ybf}+q\cdot \zbf)$ that is considered by the algorithm. In particular, we obtain an upper bound as a function of the parameters $\lambda, \mu$, and $\tvar$.  
\begin{lemma}
\label{bound at (wbf,(1-p)y'+pz)}
Consider the strategy profile $(\wbf,(1-q)\cdot \hat{\ybf}+q\cdot \zbf)$ with $q = \frac{1-\mu/2-\tvar}{1+\mu/2-\lambda -\tvar}$. Then, the regret of each player is no greater than $q\cdot (1-\lambda) + (1-q)\cdot \tvar = \frac{1-\mu/2-\tvar-\lambda +\mu \cdot \lambda/2+\mu \cdot \tvar}{1+\mu/2-\lambda -\tvar}$.
\end{lemma}
\begin{proof}
We start by showing that $q \in [0,1]$, i.e. it is well-defined.
By Corollary \ref{cor:vvar-rvar-bounds}, and since $\lambda>\frac{1}{2}, \vvar\geq 0$, we have $\tvar\leq \frac{1}{2}$. 
Also, since $\lambda<\mu \leq 1$, we have $1+\mu/2-\lambda -\tvar = 1 +\mu -\mu/2 -\lambda -\tvar> 1-\mu/2-\tvar\geq 0$, so $q$ is well-defined.

Now, we are ready to bound the regrets of the players under the strategy profile we consider. For the row player, recall that 
$\tvar = \rreg{(\wbf,\hat{\ybf})}$ and  $\rreg{(\wbf,\zbf)} \leq 1-\lambda$, by the definitions of $\tvar$ and $\lambda$. Hence, we have:
\begin{align*}
\rreg{(\wbf,(1-q)\cdot \hat{\ybf}+q\cdot \zbf)} & 
\leq (1-q)\cdot \rreg{(\wbf,\hat{\ybf})} + q \cdot \rreg{(\wbf,\zbf)} \\
& \leq (1-q)\cdot  \tvar + q \cdot (1-\lambda) \\
& = \frac{(\mu-\lambda)\cdot \tvar + (1-\mu/2-\tvar)\cdot (1-\lambda)}{1+\mu/2-\lambda -\tvar} \quad (\text{by definition of $q$})\\
& = \frac{1-\mu/2-\tvar-\lambda +\mu \cdot \lambda/2+\mu \cdot \tvar}{1+\mu/2-\lambda -\tvar}, 
\end{align*}

In order to bound the regret of the column player, recall that $\creg{(\wbf,\hat{\ybf})}\leq 1-\frac{\mu}{2}$ by Lemma \ref{regrets of strategy profile} and $\creg{(\wbf,\zbf)}\leq 1-\mu$ by the definition of $\mu$.
So, we have that
\begin{align*}
\creg{(\wbf,(1-q)\cdot \hat{\ybf}+q\cdot \zbf)}
& = (1-q)\cdot \creg{(\wbf,\hat{\ybf})} + q \cdot \creg{(\wbf,\zbf)} \\
& \leq (1-q) \cdot (1-\mu/2) + q\cdot (1-\mu)\\
& = \frac{(\mu-\lambda)\cdot (1-\mu/2) + (1-\mu/2-\tvar)\cdot (1-\mu)}{1+\mu/2-\lambda -\tvar}\\
& = \frac{1-\mu/2-\tvar+\tvar \cdot \mu -\lambda + \lambda\cdot \mu/2}{1+\mu/2-\lambda -\tvar}
\end{align*}
\end{proof}

We now come to the core of the proof and establish that 
either the $\delta$-stationary point, or the strategy profile $(\wbf,(1-q)\cdot \hat{\ybf}+q\cdot \zbf)$ yields a good approximation. This is established by the following lemma.
\begin{lemma}
\label{lem:4.2}
Under the assumptions of Case 4.2, either $(\xbfs,\ybfs)$ is a $(\frac{1}{3}+\delta)$-Nash equilibrium, or $(\wbf,(1-q)\cdot \hat{\ybf}+q\cdot \zbf)$ with $q = \frac{1-\mu/2-\tvar}{1+\mu/2-\lambda -\tvar}$, is a $\frac{1}{3}$-Nash equilibrium.
\end{lemma}
\begin{proof}

Since we are in Case 4.2, where either $\vvar+\tvar<\frac{\mu-\lambda}{2}$, or $\hat{\mu}<\mu-\vvar-\tvar$, we will split the analysis into further subcases, so that we have a more concrete relation between the relevant parameters in each subcase. More precisely, we will consider the following three subcases.
\begin{enumerate}
    \item[ \quad ] 4.2(i) $\vvar+\tvar < \frac{\mu-\lambda}{2}$.
    \item[ ] 4.2(ii) $\hat{\mu} < \mu -\vvar-\tvar$, and $\hat{\mu} \geq \vvar$.
    \item[ ] 4.2(iii) $\hat{\mu} < \mu -\vvar-\tvar$, and $\hat{\mu} < \vvar$.
\end{enumerate}
\medskip
So far, we have not been able to have a unifying argument for all these different subcases. Consequently, we proceed with a separate analysis for each of them. 

\paragraph*{Subcase 4.2(i)}

%%%%%%%%%%%%%%%%%%%%%%%%%%%%%%%%%%%%%%%%%%%%%%%%%%%%%%%%

By Lemma \ref{lem:TS-lambda-mu-bound}, the maximum regret bound for any $\delta$-stationary point is $\frac{\lambda\cdot \mu}{\lambda + \mu}+\delta$. In addition, from Lemma~\ref{bound at (wbf,(1-p)y'+pz)} the maximum regret for the strategy profile $(\wbf,(1-q)\cdot \hat{\ybf}+q\cdot \zbf)$ is bounded by $$q\cdot (1-\lambda) + (1-q)\cdot \tvar=\frac{1-\mu/2-\tvar-\lambda +\mu \cdot  \lambda/2+\mu \cdot \tvar}{1+\mu/2-\lambda -\tvar}.$$
For the sake of contradiction, assume that the regret bound at $(\xbfs,\ybfs)$ is strictly greater than $\frac{1}{3}+\delta$ and that the bound at the second profile is strictly greater than $\frac{1}{3}$.
The first assumption yields 
$$\frac{\lambda \cdot \mu}{\lambda +\mu} >\frac{1}{3} \Rightarrow \mu \cdot \lambda -\lambda/3 - \mu /3 >0 \Rightarrow \lambda > \frac{\mu}{3\mu-1}.$$
Note that $3\mu-1>0$. From the second assumption, using Lemma \ref{bound at (wbf,(1-p)y'+pz)}, we have that 
\[\frac{(1-\mu/2-\tvar) \cdot (1-\lambda) + (\mu-\lambda)\cdot \tvar}{1+\mu/2-\lambda-\tvar} > \frac{1}{3} \Rightarrow \]
\[(1-\frac{\mu}{2}-\tvar)\cdot (1-\lambda) + (\mu-\lambda)\cdot\tvar  -\frac{1}{3}-\frac{\mu}{6}+\frac{\lambda}{3}+ \frac{\tvar}{3} > 0 \Rightarrow\]
\begin{equation}
\label{eq:v+t<m-l/2}
    \frac{\lambda \mu}{2} - \frac{2\mu}{3} - \frac{2\tvar}{3} - \frac{2\lambda}{3} + \mu \cdot \tvar + \frac{2}{3} >0.
\end{equation}
To obtain a contradiction, we will establish an upper bound for the LHS of \eqref{eq:v+t<m-l/2}. Since we are in the subcase where $\tvar +\vvar < \frac{\mu-\lambda}{2}$, this implies that $\tvar< \frac{\mu-\lambda}{2}$, as $\vvar\geq 0$. Combined with the fact that $\mu>2/3$, the LHS of \eqref{eq:v+t<m-l/2} is upper bounded as follows:

\begin{multline*}
\frac{\lambda \cdot \mu}{2} - \frac{2\mu}{3}  - \frac{2\lambda}{3} + \left(\mu -\frac{2}{3}\right) \cdot \tvar + \frac{2}{3}
\leq \frac{\lambda \cdot \mu}{2} - \frac{2\mu}{3}  - \frac{2\lambda}{3} + \left(\mu -\frac{2}{3}\right) \cdot \frac{(\mu - \lambda)}{2} + \frac{2}{3}.  
\end{multline*}

After expanding the terms in the product and simplifying, the above upper bound equals $\frac{\mu^2}{2} - \mu - \frac{\lambda}{3} + \frac{2}{3}$. Since $\lambda > \frac{\mu}{3\mu-1}$, we finally have that

\[\frac{\mu^2}{2} - \mu - \frac{\lambda}{3} + \frac{2}{3} < \frac{\mu^2}{2} - \mu - \frac{\mu}{3\cdot(3\mu-1)} + \frac{2}{3} =  
\frac{(3\mu - 2)^2\cdot(\mu - 1)}{6\cdot(3\mu - 1)} \leq 0,\]
The last inequality follows since $(3\mu - 2)^2 > 0$, $\mu \leq 1$, and $(3\mu - 1) > 0$. Together with \eqref{eq:v+t<m-l/2}, we have a contradiction.
\qed

%%%%%%%%%%%%%%%%%%%%%%%%%%%%%%%%%%%%%%%%%%%%%%%%%%%%%%%%%
\paragraph*{Subcase 4.2(ii)}

By Lemma \ref{lem:stationary_bounds}, the regret of the $\delta$-stationary point is
\begin{align*}
g(\xbfs,\ybfs) & \leq \min\{P\cdot \lambda, (1-P)\cdot\mu, P \cdot \vvar + (1-P)\cdot \hat{\mu}\} +\delta \\
& \leq \min\Big\{\frac{\lambda \cdot  \mu}{\lambda + \mu},\frac{\hat{\mu}\cdot \lambda}{(\lambda +\hat{\mu}-\vvar)}\Big\} +\delta  \quad (\text{by Lemmas \ref{lem:TS-lambda-mu-bound} and \ref{lem:stationary_hatm_1}})\\
& \leq \min\Big\{\frac{\lambda \cdot \mu}{\lambda + \mu},\frac{(\mu-\vvar-\tvar)\cdot \lambda}{(\lambda +\mu -\tvar -2\vvar)}\Big\} +\delta
\end{align*}
The third inequality holds because $\frac{\hat{\mu}\cdot \lambda}{\lambda +\hat{\mu}-\vvar}$ is an increasing function of $\hat{\mu}$ (this can be verified by taking the derivative and then using the fact that $\lambda> \vvar$), and $\hat{\mu} < \mu-\vvar-\tvar$.  

\medskip
\noindent
Now we consider two cases, in terms of $\mu-\lambda$.
\begin{itemize}
    \item If $\mu-\lambda < \tvar$, we have that  $\hat{\mu} < \mu -\vvar - \mu +\lambda = \lambda - \vvar$.
    So, using the inequality we derived above, we get that the approximation bound of a $\delta$-stationary point in this case is $\frac{\hat{\mu} \cdot \lambda}{\lambda+\hat{\mu}-\vvar}+\delta\leq \frac{\lambda \cdot (\lambda-\vvar)}{2\lambda-2\vvar}+\delta = \frac{\lambda}{2}+\delta\leq \frac{1}{3}+\delta$.
    \item If $\mu-\lambda \geq \tvar$, then $\frac{(\mu-\vvar-\tvar)\cdot \lambda}{\lambda+\mu-\tvar-2\vvar}$ is increasing with respect to $\vvar$ (it can be easily verified by looking at the derivative). Therefore, using Corollary \ref{cor:vvar-rvar-bounds}, that $\vvar \leq 1-\lambda-\tvar$, we get that
    \begin{align*}
g(\xbfs,\ybfs) 
& \leq \min\Big\{\frac{\lambda \cdot \mu}{\lambda + \mu},\frac{(\mu-\vvar-\tvar)\cdot \lambda}{(\lambda +\mu -\tvar -2\vvar)}\Big\} +\delta\\
& \leq \min\Big\{\frac{\lambda \cdot \mu}{\lambda + \mu},\frac{\lambda \cdot (\lambda + \mu - 1))}{(3\lambda + \mu + \tvar - 2)}\Big\}+\delta. \end{align*}
\end{itemize}
Hence, given the analysis above, in what follows we will assume that $\mu - \lambda \geq t_r$, since otherwise we have a $(\frac{1}{3}+\delta)$-Nash equilibrium.

\medskip
\noindent
Assume now that the approximation bound of $(\xbfs,\ybfs)$ is worse than $\frac{1}{3}+\delta$.  
Thus, we have $\frac{\lambda \cdot \mu}{\lambda + \mu} >\frac{1}{3}$ and $\frac{\lambda \cdot (\lambda + \mu - 1))}{(3 \lambda + \mu + \tvar - 2)}>\frac{1}{3}$. 
We will prove by contradiction that in this case the approximation bound of $(\wbf,(1-q)\cdot \hat{\ybf}+q\cdot \zbf)$ is at most $\frac{1}{3}$.
For the sake of contradiction, assume that the approximation bound of $(\wbf,(1-q)\cdot \hat{\ybf}+q\cdot \zbf)$ is strictly worse than $\frac{1}{3}$, i.e., using Lemma \ref{bound at (wbf,(1-p)y'+pz)}, we assume that 
$\frac{(1-\mu/2-\tvar) \cdot (1-\lambda) + (\mu-\lambda)\cdot \tvar}{1+\mu/2-\lambda-\tvar}>1/3$.

Next, we will use the three inequalities from above in order to derive our contradiction. From the first inequality, i.e. $\frac{\lambda\cdot \mu}{\lambda+\mu} > \frac{1}{3}$, we get that
\begin{align}
\label{eq:aux1}
    \mu > \frac{\lambda}{3\lambda-1}.
\end{align}
From the second inequality, i.e. $\frac{\lambda \cdot (\lambda + \mu - 1))}{(3 \lambda + \mu + \tvar - 2)}>\frac{1}{3}$, we get $\lambda\cdot(\lambda + \mu - 1) - \frac{\mu}{3} - \frac{\tvar}{3} - \lambda + \frac{2}{3} >0$, which in turn implies that  
\begin{equation} 
\label{eq:r_bound}
0\leq \tvar< 3\lambda\cdot (\lambda + \mu - 1) - \mu - 3\cdot\lambda + 2.
\end{equation}

From the third inequality, i.e., from the regret bound of the profile $(\wbf,(1-p)\cdot \hat{\ybf}+p\cdot \zbf)$, we obtain 
\[\frac{(1-\mu/2-\tvar) \cdot (1-\lambda) + (\mu-\lambda) \cdot \tvar}{1+\mu/2-\lambda-\tvar} > \frac{1}{3} \Rightarrow \]
\[\left(1-\frac{\mu}{2}-\tvar\right)\cdot (1-\lambda) + (\mu-\lambda)\cdot \tvar  -\frac{1}{3}-\frac{\mu}{6}+\frac{\lambda}{3}+\frac{\tvar}{3}>0 \Rightarrow\]
\begin{align}
\label{eq:aux3}
\frac{\lambda\cdot \mu}{2} - \frac{2\mu}{3}-\frac{2\lambda}{3}+\left(\mu - \frac{2}{3} \right) \cdot \tvar+\frac{2}{3}>0 
\end{align}
We will prove that it is not possible that all Inequalities~\eqref{eq:aux1}--\eqref{eq:aux3} simultaneously hold.

\medskip
\noindent
Let us focus on Inequality~\eqref{eq:aux3}. We will prove that if \eqref{eq:aux1} and \eqref{eq:r_bound} hold, then the left hand side (LHS) of \eqref{eq:aux3} {\em cannot} be positive. To this end, we will upper bound the LHS of \eqref{eq:aux3}.
Observe that by our assumption that $\mu> \frac{2}{3}$, the term $(\mu-2/3)\cdot \tvar$ in \eqref{eq:aux3} is increasing with $\tvar$. So, if we use the upper bound for $\tvar$ from~\eqref{eq:r_bound} (and after simplifying the resulting expressions), we get that the LHS of \eqref{eq:aux3} is upper bounded by
\begin{align}
    \label{eq:aux4}
    (3\lambda -1)\cdot\mu^2+(3\lambda^2-\frac{15\lambda}{2}+2)\cdot\mu  -2\lambda^2+\frac{10\lambda}{3} -\frac{2}{3}
\end{align}

Let us view \eqref{eq:aux4} as a function of $\mu$. Since $\lambda > 1/2$ we get that $3\lambda-1>0$, which implies that for any value of $\lambda\in(1/2, 2/3]$, we have a quadratic function of $\mu$ whose second derivative is positive. This implies that the maximum of the function within any interval will be achieved at one of its endpoints. 
By \eqref{eq:aux1}, we know that for any $\lambda$, the value of $\mu$ ranges in $(\frac{\lambda}{3\lambda-1}, 1]$. Hence, for any $\lambda$, the function defined in \eqref{eq:aux4} is upper bounded either by its value at $\mu=\frac{\lambda}{3\lambda-1}$ or at $\mu=1$.    

We can continue now as follows.
\begin{itemize}
    \item If we set $\mu = \frac{\lambda}{3\cdot\lambda -1}$ in  \eqref{eq:aux4}, we get an upper bound of 
    \[\frac{-((2\lambda - 1) \cdot (3\lambda - 2)^2)}{6\cdot(3\lambda - 1)}.\]
    Observe that this quantity is less than or equal to zero since $\lambda > \frac{1}{2}$ and thus: $2\lambda - 1> 0$; $(3\lambda - 2)^2\geq 0$; and $3\lambda - 1> 0$. 
    \item When $\mu=1$ in \eqref{eq:aux4}, the resulting expression is
    $$\lambda^2 - \frac{7\lambda}{6} + \frac{1}{3}.$$
    But again, it can be verified that this quantity is non-positive for any $\lambda \in (\frac{1}{2},\frac{2}{3}]$. In fact, and quite surprisingly, this is the only interval where this specific polynomial takes negative values.
\end{itemize}
Thus, in both of the above cases we get a non-positive expression, implying that the LHS of \eqref{eq:aux3} is non-positive, which is a contradiction.
\qed
%%%%%%%%%%%%%%%%%%%%%%%%%%%%%%%%%%%%%%%%%%%%%%%%%%%%%%%%%%%

\paragraph*{Subcase 4.2(iii)}

Let us begin by observing that if $\vvar \leq \frac{1}{3}$, then the strategy profile $(\xbfs,\ybfs)$ is a $(\frac{1}{3}+\delta)$-NE. 
Indeed, from Lemma~\ref{lem:stationary_hatm2}, and since $\hat{\mu}<\vvar$, we get that the approximation guarantee of the stationary strategy profile is $\frac{\vvar \cdot \mu}{\mu-\hat{\mu}+\vvar}+\delta \leq \vvar+\delta\leq \frac{1}{3}+\delta$.  Hence, in what follows, we will assume that $\vvar>\frac{1}{3}$.
So, from Corollary \ref{cor:vvar-rvar-bounds} we have that 
\begin{align}
\label{eq:aux2-1}
    \tvar \leq \frac{1-\lambda - \vvar}{2} < \frac{\frac{2}{3} - \lambda}{2}= \frac{1}{3} - \frac{\lambda}{2}.
\end{align}

\noindent Assume now for the sake of contradiction that $(\wbf,(1-q)\cdot \hat{\ybf}+q\cdot \zbf)$ is not a $\frac{1}{3}$-NE, i.e., the maximum regret is higher than $1/3$. Combining this with the bound of Lemma~\ref{bound at (wbf,(1-p)y'+pz)} on the regret, and by simplifying the resulting expression, 
we get the following inequality:
\begin{align}
    \label{eq:aux2-3}
    \frac{\lambda \mu}{2} - \frac{2\mu}{3}-\frac{2\lambda}{3}+\big(\mu-\frac{2}{3}\big)\cdot\tvar+\frac{2}{3}>0.
\end{align}
We will focus on the LHS of~\eqref{eq:aux2-3} and we will upper bound it by a non-positive value. 
Now, assume that $(\xbfs, \ybfs)$ is not a $(\frac{1}{3}+\delta)$-NE; if it was then our algorithm would return this strategy profile. This means from Lemma~\ref{lem:TS-lambda-mu-bound} that $\frac{\lambda\cdot \mu}{\lambda + \mu} > \frac{1}{3}$, which implies that $\mu > \frac{\lambda}{3\lambda-1}$.
In addition, observe that since $\mu > 2/3$, the LHS of~\eqref{eq:aux2-3} is increasing with $\tvar$. So, if we use Inequality~\eqref{eq:aux2-1}, we get that the LHS of~\eqref{eq:aux2-3} is upper bounded by
\begin{align*}
    \frac{\lambda \mu}{2} - \frac{2\mu}{3}-\frac{2\lambda}{3}+(\mu-\frac{2}{3}) \cdot (\frac{1}{3} - \frac{\lambda}{2})  +\frac{2}{3} 
    & = -\frac{\mu}{3} -\frac{\lambda}{3} + \frac{4}{9}\\
    & < - \frac{\lambda}{9\lambda-3} -\frac{\lambda}{3} + \frac{4}{9} \quad \left(\text{since $\mu > \frac{\lambda}{3\lambda-1}$} \right).
\end{align*}
Observe though that this quantity is non-positive for every $\lambda \in [\frac{1}{2}, \frac{2}{3}]$, which in turn contradicts Inequality~\eqref{eq:aux2-3}.
\end{proof}

\section{Discussion}
Our algorithm is the first improvement for a foundational problem after 15 years, during which progress had stalled. 
We hope that our result will again ignite the spark for actively studying \eps-NE in bimatrix games. There is still a large gap between the quasi polynomial-time lower bound for ``some'' very small constant $\eps^{\star}$ from~\cite{Rub16-PETH}, and our newly-established upper bound of $1/3$. 
We conjecture that closing this gap requires radically new ideas.

Our result has some extra positive consequences for games with more than two players.
In~\cite{BBM10} it was shown that if we have an algorithm that finds an $\alpha$-Nash equilibrium in a $(k-1)$-player game, then in polynomial time we can compute a $(\frac{1}{2-\alpha})$-NE for any $k$-player game.
Thus, our algorithm improves the state of the art for $k$-player normal-form games, for any $k>2$. 
Namely, we get $(0.6+\delta)$-NE for three-player games, $(5/7+\delta)$-NE for four-player games, and so on.

\paragraph{\bf Acknowledgements.} We would like to thank Hanyu Li for spotting an issue with one of our proofs in the first version of our paper.

\bibliographystyle{plain}
\bibliography{agt}

\begin{thebibliography}{10}

\bibitem{adsul2011rank1}
Bharat Adsul, Jugal Garg, Ruta Mehta, and Milind Sohoni.
\newblock Rank-1 bimatrix games: a homeomorphism and a polynomial time
  algorithm.
\newblock In {\em Proceedings of {STOC}}, pages 195--204, 2011.

\bibitem{austrin2011inapproximability}
Per Austrin, Mark Braverman, and Eden Chlamt{\'a}{\v{c}}.
\newblock Inapproximability of np-complete variants of {N}ash equilibrium.
\newblock In {\em Approximation, Randomization, and Combinatorial Optimization.
  Algorithms and Techniques}, pages 13--25. Springer, 2011.

\bibitem{BabichenkoBP17}
Yakov Babichenko, Siddharth Barman, and Ron Peretz.
\newblock Empirical distribution of equilibrium play and its testing
  application.
\newblock {\em Math. Oper. Res.}, 42(1):15--29, 2017.

\bibitem{barany2007nash-random}
Imre B{\'a}r{\'a}ny, Santosh Vempala, and Adrian Vetta.
\newblock {N}ash equilibria in random games.
\newblock {\em Random Structures \& Algorithms}, 31(4):391--405, 2007.

\bibitem{Barman18-caratheodory-qptas}
Siddharth Barman.
\newblock Approximating {N}ash equilibria and dense subgraphs via an
  approximate version of {C}arath{\'{e}}odory's theorem.
\newblock {\em {SIAM} J. Comput.}, 47(3):960--981, 2018.

\bibitem{BoodaghiansBHR20-Smoothed-2Nash}
Shant Boodaghians, Joshua Brakensiek, Samuel~B. Hopkins, and Aviad Rubinstein.
\newblock Smoothed complexity of 2-player {N}ash equilibria.
\newblock In {\em Proceedings of {FOCS}}, pages 271--282, 2020.

\bibitem{BBM10}
Hartwig Bosse, Jaroslaw Byrka, and Evangelos Markakis.
\newblock New algorithms for approximate {N}ash equilibria in bimatrix games.
\newblock {\em Theoretical Computer Science}, 411(1):164--173, 2010.

\bibitem{BravermanKW15-QP-LB}
Mark Braverman, Young Kun{-}Ko, and Omri Weinstein.
\newblock Approximating the best {N}ash equilibrium in
  \emph{n\({}^{\mbox{o}}\)}\({}^{\mbox{(log \emph{n})}}\)-time breaks the
  exponential time hypothesis.
\newblock In {\em Proceedings of {SODA}}, pages 970--982. {SIAM}, 2015.

\bibitem{chen2006-sparse}
Xi~Chen, Xiaotie Deng, and Shang-Hua Teng.
\newblock Sparse games are hard.
\newblock In {\em Proceedings of {WINE}}, pages 262--273, 2006.

\bibitem{CDT09}
Xi~Chen, Xiaotie Deng, and Shang-Hua Teng.
\newblock Settling the complexity of computing two-player {N}ash equilibria.
\newblock {\em Journal of the ACM}, 56(3), 2009.

\bibitem{chen2007-win-lose}
Xi~Chen, Shang-Hua Teng, and Paul Valiant.
\newblock The approximation complexity of win-lose games.
\newblock In {\em Proceedings of {SODA}}, volume~7, pages 159--168, 2007.

\bibitem{ChenDHLL21-Deng-TS-tight}
Zhaohua Chen, Xiaotie Deng, Wenhan Huang, Hanyu Li, and Yuhao Li.
\newblock On tightness of the {T}saknakis-{S}pirakis algorithm for approximate
  {N}ash equilibrium.
\newblock In {\em Proceedings of {SAGT}}, volume 12885, pages 97--111, 2021.

\bibitem{codenotti2005-win-lose}
Bruno Codenotti and Daniel {\v{S}}tefankovi{\v{c}}.
\newblock On the computational complexity of {N}ash equilibria for (0,1)
  bimatrix games.
\newblock {\em Information Processing Letters}, 94(3):145--150, 2005.

\bibitem{czumaj2019distributed-wsne}
Artur Czumaj, Argyrios Deligkas, Michail Fasoulakis, John Fearnley, Marcin
  Jurdzi{\'n}ski, and Rahul Savani.
\newblock Distributed methods for computing approximate equilibria.
\newblock {\em Algorithmica}, 81(3):1205--1231, 2019.

\bibitem{CFJ14}
Artur Czumaj, Michail Fasoulakis, and Marcin Jurdzi\'nski.
\newblock Approximate well-supported {N}ash equilibria in symmetric bimatrix
  games.
\newblock In {\em Proceedings of {SAGT}}, volume 8768, pages 244--254, 2014.

\bibitem{DGP06}
Constantinos Daskalakis, Paul~W. Goldberg, and Christos~H. Papadimitriou.
\newblock The complexity of computing a {N}ash equilibrium.
\newblock In {\em Proceedings of {STOC}}, pages 71--78, 2006.

\bibitem{DMP07}
Constantinos Daskalakis, Aranyak Mehta, and Christos Papadimitriou.
\newblock Progress in approximate {N}ash equilibria.
\newblock In {\em Proceedings of {EC}}, pages 355--358, 2007.

\bibitem{DMP09}
Constantinos Daskalakis, Aranyak Mehta, and Christos Papadimitriou.
\newblock A note on approximate {N}ash equilibria.
\newblock {\em Theoretical Computer Science}, 410(17):1581--1588, 2009.

\bibitem{DeligkasFMS22-eps-ETR}
Argyrios Deligkas, John Fearnley, Themistoklis Melissourgos, and Paul~G.
  Spirakis.
\newblock Approximating the existential theory of the reals.
\newblock {\em J. Comput. Syst. Sci.}, 125:106--128, 2022.

\bibitem{DeligkasFS18-constrained-QP-LB}
Argyrios Deligkas, John Fearnley, and Rahul Savani.
\newblock Inapproximability results for constrained approximate {N}ash
  equilibria.
\newblock {\em Inf. Comput.}, 262:40--56, 2018.

\bibitem{DFSS17}
Argyrios Deligkas, John Fearnley, Rahul Savani, and Paul~G. Spirakis.
\newblock Computing approximate {N}ash equilibria in polymatrix games.
\newblock {\em Algorithmica}, 77(2):487--514, 2017.

\bibitem{FGSS16-approximate-WSNE}
John Fearnley, Paul~W. Goldberg, Rahul Savani, and Troels~Bjerre S{\o}rensen.
\newblock Approximate well-supported {N}ash equilibria below two-thirds.
\newblock {\em Algorithmica}, 76(2):297--319, 2016.

\bibitem{kannan2010-rank-games}
Ravi Kannan and Thorsten Theobald.
\newblock Games of fixed rank: A hierarchy of bimatrix games.
\newblock {\em Economic Theory}, 42(1):157--173, 2010.

\bibitem{KPS06}
Spyros~C. Kontogiannis, Panagiota~N. Panagopoulou, and Paul~G. Spirakis.
\newblock Polynomial algorithms for approximating {N}ash equilibria of bimatrix
  games.
\newblock In {\em Proceedings of {WINE}}, pages 286--296, 2006.

\bibitem{KS-wsne}
Spyros~C. Kontogiannis and Paul~G. Spirakis.
\newblock Well supported approximate equilibria in bimatrix games.
\newblock {\em Algorithmica}, 57(4):653--667, 2010.

\bibitem{KS11}
Spyros~C. Kontogiannis and Paul~G. Spirakis.
\newblock Approximability of symmetric bimatrix games and related experiments.
\newblock In {\em Proceedings of SEA}, volume 6630, pages 1--20, 2011.

\bibitem{KothariM18-sum-of-squares}
Pravesh~K. Kothari and Ruta Mehta.
\newblock Sum-of-squares meets {N}ash: lower bounds for finding any
  equilibrium.
\newblock In {\em Proceedings of {STOC}}, pages 1241--1248. {ACM}, 2018.

\bibitem{LMM03}
Richard Lipton, Evangelos Markakis, and Aranyak Mehta.
\newblock Playing large games using simple strategies.
\newblock In {\em Proceedings of {EC}}, pages 36--41, 2003.

\bibitem{liu2018-win-lose}
Zhengyang Liu and Ying Sheng.
\newblock On the approximation of {N}ash equilibria in sparse win-lose games.
\newblock In {\em Proceedings of {AAAI}}, volume~32, 2018.

\bibitem{mclennan2010-imitation}
Andrew McLennan and Rabee Tourky.
\newblock Imitation games and computation.
\newblock {\em Games and Economic Behavior}, 70(1):4--11, 2010.

\bibitem{mclennan2010simple-imitation}
Andrew McLennan and Rabee Tourky.
\newblock Simple complexity from imitation games.
\newblock {\em Games and Economic Behavior}, 68(2):683--688, 2010.

\bibitem{Mehta18-constant-rank}
Ruta Mehta.
\newblock Constant rank two-player games are {PPAD}-hard.
\newblock {\em {SIAM} J. Comput.}, 47(5):1858--1887, 2018.

\bibitem{MurhekarMehta20-imitation}
Aniket Murhekar and Ruta Mehta.
\newblock Approximate {N}ash equilibria of imitation games: Algorithms and
  complexity.
\newblock In {\em Proceedings of {AAMAS}}, pages 887--894, 2020.

\bibitem{Nas51}
John Nash.
\newblock Non-cooperative games.
\newblock {\em Annals of Mathematics}, 54(2):286--295, 1951.

\bibitem{panagopoulou2014random}
Panagiota~N. Panagopoulou and Paul~G. Spirakis.
\newblock Random bimatrix games are asymptotically easy to solve (a simple
  proof).
\newblock {\em Theory of Computing Systems}, 54(3):479--490, 2014.

\bibitem{Rub16-PETH}
Aviad Rubinstein.
\newblock Settling the complexity of computing approximate two-player {N}ash
  equilibria.
\newblock In {\em Proceedings of {FOCS}}, pages 258--265, 2016.

\bibitem{TS08}
Haralambos Tsaknakis and Paul~G. Spirakis.
\newblock An optimization approach for approximate {N}ash equilibria.
\newblock {\em Internet Mathematics}, 5(4):365--382, 2008.

\end{thebibliography}

%===============================================================

\end{document}